\documentclass{llncs}

\usepackage{amsmath}
\usepackage{txfonts}
\usepackage{xspace}
\usepackage{url}

\newcommand{\pdl}{\logic{PDL}}
\newcommand{\dlpa}{\logic{DL\text{-}PA}}

\newcommand{\branch}{b}
\newcommand{\Branches}{B}
\newcommand{\ecl}{\cl^+}
\newcommand{\cl}{\operatorname{cl}}
\newcommand{\clp}{\cl^{\Box}}
\newcommand{\Tableau}{T}

\newcommand{\Function}[1]{\textbf{function}~#1}
\newcommand{\Input}[1]{\textbf{input:}~#1}
\newcommand{\Output}[1]{\textbf{output:}~#1}
\newcommand{\Begin}{\textbf{begin}}
\newcommand{\End}{\textbf{end}}
\newcommand{\Return}[1]{\textbf{return}~#1}

\newcommand{\If}[1]{\textbf{if}~#1~\textbf{then}}

\newcommand{\ElseIf}[1]{\textbf{else if}~#1~\textbf{then}}
\newcommand{\Else}{\textbf{else}}
\newcommand{\EndIf}{\textbf{end if}}
\newcommand{\While}[1]{\textbf{while}~#1~\textbf{do}}
\newcommand{\EndWhile}{\textbf{end while}}

\newcommand{\ForEach}[1]{\textbf{for each}~#1~\textbf{do}}
\newcommand{\EndFor}{\textbf{end for}}
\newcommand{\True}{\textbf{true}}
\newcommand{\False}{\textbf{false}}

\newcommand{\Acts}{\mathbb{A}}
\newcommand{\act}{a}
\newcommand{\card}{\operatorname{card}}
\newcommand{\dom}{\operatorname{dom}}
\newcommand{\emptyseq}{()}
\newcommand{\exto}[1]{|| #1 ||}			

\newcommand{\lang}{\mathcal{L}}
\newcommand{\lbox}[1]{[#1]}
\newcommand{\ldia}[1]{\langle #1 \rangle}
\newcommand{\len}{\operatorname{len}}
\newcommand{\lequiv}{\leftrightarrow}
\newcommand{\limp}{\rightarrow}
\newcommand{\logic}[1]{\ensuremath{\mathrm{#1}}\xspace}
\newcommand{\Nats}{\mathbb{N}}
\newcommand{\ndet}{\cup}						
\newcommand{\prop}{p}
\newcommand{\propb}{q}
\newcommand{\Props}{\mathbb{P}}
\newcommand{\seq}{\operatorname{exe}}
\newcommand{\set}[1]{\{ #1 \}}
\newcommand{\tuple}[1]{\langle #1 \rangle}
\newcommand{\tr}{\operatorname{tr}}
\newcommand{\Vals}{V}
\renewcommand{\phi}{\varphi}

\newcounter{warning}
\newcommand{\warning}[1]{}

\newcommand{\onlylong}[1]{#1}
\newcommand{\onlyshort}[1]{}

\newenvironment{relemma}[2]{\par\noindent\textbf{Lemma~#1~(#2).~~}}{\par\par}

\newenvironment{retheorem}[2]{\par\noindent\textbf{Theorem~#1~(#2).~~}}{\par\par}

\begin{document}
\title{Tableaux for \\ Dynamic Logic of Propositional Assignments}
\titlerunning{Tableaux for $\dlpa$}  
%
\author{%
Tiago de Lima\inst{1}
\and
Andreas Herzig\inst{2}
}
\authorrunning{T. de Lima and A. Herzig} 
%
%
\institute{%
CRIL, Univ. of Artois and CNRS,
Rue Jean Souvraz, SP 18,
62307 Lens Cedex,
France
\and
IRIT, Univ. of Toulouse and CNRS,
118, route de Narbonne,
31062 Toulouse Cedex 9,
France
}
\maketitle
\pagestyle{plain}
\begin{abstract}
The Dynamic Logic for Propositional Assignments ($\dlpa$) has recently been studied as an alternative to Propositional Dynamic Logic ($\pdl$).
In $\dlpa$, the abstract atomic programs of $\pdl$ are replaced by assignments of propositional variables to truth values.
This makes $\dlpa$ enjoy some interesting meta-logical properties that $\pdl$ does not, such as eliminability of the Kleene star, compactness and interpolation.
We define and analytic tableaux calculus for $\dlpa$ and show that it matches the known complexity results.
\end{abstract}
%
\section{Introduction}
\label{sec:intro}
%
Dynamic Logic of Propositional Assignments, abbreviated $\dlpa$, has recently been studied in
\cite{HerzigLMT-Ijcai11,BalbianiHerzigTroquard-Lics13}
as an interesting alternative to Propositional Dynamic Logic ($\pdl$) \cite{Harel84}.
In a series of papers, it was shown that $\dlpa$ is a useful tool to analyse various kinds of dynamic systems, such as normative systems \cite{HerzigLMT-Ijcai11}, fusion operators \cite{HerzigEtal-Foiks14}, update and revision operators \cite{HerzigEtal-Kr14}, or the evolution of argumentation frameworks \cite{DoutreEtal-Kr14}.
While, in $\pdl$, one can write formulas of the form $\lbox{a}\phi$, meaning ``after every possible execution of the abstract atomic program $\act$, formula $\phi$ is true'', in $\dlpa$, one can write formulas of the form $\lbox{{+}\prop}\phi$, meaning ``after assigning the truth value of $\prop$ to true, formula $\phi$ is true''.
The atomic program ${+}\prop$ is an assignment that maps the propositional variable $\prop$ to true.
In fact, the atomic programs of $\dlpa$ are sets of such assignments, that are viewed as partial functions from the set of propositional variables to $\set{\top,\bot}$.
From these atomic programs, complex programs are built just as in $\pdl$.
For example, one can write in $\dlpa$ the formula $\lbox{\lnot\prop?; {+}\prop}\prop$, which means
``if $\prop$ is false then $\prop$ is true after its truth value be assigned to true''.

While the models of $\pdl$ are transition systems, the models of $\dlpa$ are much smaller: 
valuations of classical propositional logic, i.e., nothing but sets of propositional variables.
Due to that, $\dlpa$ enjoys some interesting meta-logical properties that $\pdl$ lacks, such as eliminability of the Kleene star, compactness and interpolation.

The complexity of the satisfiability problem in $\dlpa$ is the same as in $\pdl$:
it is in EXPTIME for the full language and PSPACE complete for the star-free fragment.
EXPTIME membership of the full language is proved by a polynomial embedding of $\dlpa$ into $\pdl$, and PSPACE membership of the star-free fragment of $\dlpa$ is proved via NPSPACE membership of its model checking problem, exploiting the fact that NPSPACE = PSPACE. 
However, these reductions lead to suboptimal theorem proving methods.
Our aim in this paper is to define tableaux theorem proving procedures for $\dlpa$ that are both direct and more efficient.

The paper is organized as follows.
We start by recalling $\dlpa$ (Section~\ref{sec:dlpa}).
Then, we provide a tableaux method for its star-free fragment (Section~\ref{sec:tableaux_star_free}) and show an algorithm implementing it that works in polynomial space (Section \ref{sec:optimal_star_free}).
After that, we extend the tableaux method to the full language of $\dlpa$ (Section~\ref{sec:tableaux}) and then show an algorithm implementing it and that works in time exponential (Section \ref{sec:optimal}).
Section \ref{sec:conclu} discusses some issues and concludes the paper.\onlyshort{\footnotemark
\footnotetext{%
Proofs of the important results are in the appendix of the long version, which can be found at:\\
\url{https://drive.google.com/file/d/0B2ruhCOtksfObFZDS2V6dDZHSjQ/edit?usp=sharing}.
}
}\onlylong{\footnotemark
\footnotetext{%
Proofs of the important theorems are in the appendix.
}
}
%
\section{Dynamic Logic of Propositional Assignments}
\label{sec:dlpa}
%
\subsection{Syntax}
The vocabulary of $\dlpa$ contains a countable set $\Props$ of propositional variables.
From this set, we build the set $\Acts$ of \emph{propositional assignments}, which are the \emph{atomic programs} of the language.
Each propositional assignment is a non-empty finite partial function from $\Props$ to $\set{\bot, \top}$.\footnotemark
\footnotetext{%
We note that the original language in \cite{BalbianiHerzigTroquard-Lics13} is slightly more restrictive:
$\alpha$ only assigns a single propositional variable.
But, as shown in this paper, it does not change the known decidability and complexity results.
}

The \emph{language $\lang$ of $\dlpa$} is the the set of formulas $\phi$ defined by the BNF:
\begin{align*}
\phi & \Coloneqq \prop \mid \lnot\phi \mid \phi \land \phi \mid \lbox{\pi}\phi\\
\pi & \Coloneqq \alpha \mid \pi; \pi \mid \pi \ndet \pi \mid \pi^\ast \mid \phi?
\end{align*}
where $\prop$ ranges over $\Props$ and $\alpha$ ranges over $\Acts$.

To ease notation and readability of programs, we write ${+}\prop$ for $(\prop, \top)$ and ${-}\prop$ for $(\prop, \bot)$.
Moreover, we sometimes ``forget'' some parentheses and curly braces when writing propositional assignments.
As a result, the formula $\lbox{\set{(\prop, \top), (\propb, \bot)}}\phi$ is rather noted $\lbox{{+}\prop, {-}\propb}\phi$.
In some places, we use the expression $\pm\prop$ to talk economically about ${+}\prop$ and ${-}\prop$ at the same time.

The complex programs of $\dlpa$ are constructed as in Dynamic Propositional Logic ($\pdl$) \cite{HarelKozenTiuryn00}.
As well as in $\pdl$, formulas of the form $\lbox{\pi}\phi$ are read
``after every possible execution of $\pi$, $\phi$ is true''.

We also use the common abbreviations for the connectives $\top$, $\bot$, $\lor$, $\limp$ and $\lequiv$.
The formula $\ldia{\pi} \phi$ abbreviates $\lnot\lbox{\pi}\lnot\phi$. 
The \emph{star-free fragment of $\lang$} is the fragment without the Kleene star operator $^\ast$ and is noted $\lang^{-\ast}$.

The \emph{length} of a formula or a program, given by the function $\len$, is the number of atoms and connectives in the formula or the program. Table~\ref{tab:len} defines it formally.

\begin{table}
\begin{align*}
\begin{aligned}[t]
\len(\prop) & = 1\\
\len(\lnot\phi) & = 1 + \len(\phi)\\
\len(\phi_1 \land \phi_2) & = 1 + \len(\phi_1) + \len(\phi_2)\\
\len(\lbox{\pi}\phi) & = 1 + \len(\pi) + \len(\phi)
\end{aligned}
& \qquad
\begin{aligned}[t]
\len(\alpha) & = |\dom(\alpha)|\\
\len(\pi_1; \pi_2) & = 1 + \len(\pi_1) + \len(\pi_2)\\
\len(\pi_1 \cup \pi_2) & = 1 + \len(\pi_1) + \len(\pi_2)\\
\len(\pi^\ast) & = 1 + \len(\pi)\\
\len(\phi?) & = 1 + \len(\phi)
\end{aligned}
\end{align*}
\caption{%
\label{tab:len}
Length
}
\end{table}

The \emph{closure of $\phi$} is the set $\cl(\phi)$ defined in Table~\ref{tab:cl}.
This is almost the same as the Fisher-Ladner closure \cite{Fischer_Ladner-1979-JCSS}, which is used to show decidability and complexity results for $\pdl$.
But since the atomic programs of $\dlpa$ are sets of assignments, there is a difference here in the definition of $\clp(\lbox{\alpha}\phi)$.
It takes into account the assignments by adding the domain of the atomic program $\alpha$.

\begin{table}
\begin{align*}
\begin{aligned}[t]
\cl(\prop) & = \set{\prop}\\
\cl(\lnot\phi) & = \set{\lnot\phi} \cup \cl(\phi)\\
\cl(\phi_1 \land \phi_2) & = \set{\phi_1 \land \phi_2} \cup \cl(\phi_1) \cup \cl(\phi_2)\\
\cl(\lbox{\pi}\phi) & = \clp(\lbox{\pi}\phi) \cup \cl(\phi)
\end{aligned}
& \qquad
\begin{aligned}[t]
\clp(\lbox{\alpha}\phi) & = \set{\lbox{\alpha}\phi} \cup \dom(\alpha)\\
\clp(\lbox{\pi_1; \pi_2}\phi) & = \set{\lbox{\pi_1; \pi_2}\phi} \cup \clp(\lbox{\pi_1}\lbox{\pi_2}\phi)\\
\clp(\lbox{\pi_1 \ndet \pi_2}\phi) & = \set{\lbox{\pi_1 \ndet \pi_2}\phi} \cup \clp(\lbox{\pi_1}\phi) \cup \clp(\lbox{\pi_2}\phi)\\
\clp(\lbox{\pi^\ast}\phi) & = \set{\lbox{\pi^\ast}\phi} \cup \clp(\lbox{\pi}\lbox{\pi^\ast}\phi)\\
\clp(\lbox{\phi_1?}\phi_2) & = \set{\lbox{\phi_1?}\phi_2} \cup \cl(\phi_1)
\end{aligned}
\end{align*}
\caption{%
\label{tab:cl}
Closure
}
\end{table}

The \emph{extended closure of $\phi$} is the set $\ecl(\phi)$ containing $\cl(\phi)$ and the negations of its formulas, i.e.,
$\ecl(\phi) = \cl(\phi) \cup \set{\lnot\psi : \psi \in \cl(\phi)}$.
To ease notation, we sometimes use $\Props_\phi$ to denote the set of propositional variables occurring in $\phi$,
i.e. $\Props_\phi = \Props \cap \cl(\phi)$.

The lemma below can be proved with an easy induction on the length of formulas and programs.

\begin{lemma}
\label{lem:cl}
~
\begin{enumerate}
\item
$\card(\clp(\lbox{\pi}\phi)) \leq \len(\lbox{\pi}\phi)$
\item
$\card(\cl(\phi)) \leq \len(\phi)$.
\item
$\card(\ecl(\phi)) \leq 2\len(\phi)$.
\end{enumerate}
\end{lemma}

Intuitively, the set of \emph{execution traces} of $\pi$ is the set $\seq(\pi)$ of sequences of assignments that corresponds to all possible executions of program $\pi$.
The set $\seq(\phi)$ corresponds to all possible executions of all programs in $\phi$.
These are defined by a mutual recursion, as displayed in Table~\ref{tab:seq}.
We use the symbol `$\emptyseq$' to denote the empty sequence.

\begin{table}
\begin{align*}
\begin{aligned}[t]
\seq(\prop) & = \set{\emptyseq}\\
\seq(\lnot\phi) & = \seq(\phi)\\
\seq(\phi \land \psi) & = \seq(\phi) \cup \seq(\psi)\\
\seq(\lbox{\pi}\phi) & = \seq(\pi) \cup \seq(\phi)
\end{aligned}
& \qquad
\begin{aligned}[t]
\seq(\alpha) &= \set{\alpha}\\
\seq(\pi_1;\pi_2) &= \set{\sigma_1\sigma_2 : \sigma_1 \in \seq(\pi_1), \sigma_2 \in \seq(\pi_2)}\\
\seq(\pi_1 \ndet \pi_2) &= \seq(\pi_1) \cup \seq(\pi_2)\\
\seq(\pi^\ast) &= \bigcup_{n \in \Nats_0}\set{\sigma_1\dots\sigma_n : \sigma_1, \dots, \sigma_n \in \seq(\pi)}\\
\seq(\phi?)  &= \seq(\phi)
\end{aligned}
\end{align*}
\caption{%
\label{tab:seq}
Execution traces
}
\end{table}

The \emph{length of execution traces}, also given by the function $\len$, is just the number of atomic programs in it.
That is:
$\len(\emptyseq) = 0$,
$\len(\alpha) = 1$
and
$\len(\sigma\alpha) = \len(\sigma) + \len(\alpha)$.

The lemma below can also be proved with an easy induction on the length of programs and formulas:

\begin{lemma}
\label{lem:exe}
~
\begin{enumerate}
\item
If $\pi$ does not contain the Kleene star then $\len(\sigma) \leq \len(\pi)$, for all $\sigma \in \seq(\pi)$.
\item
If $\phi \in \lang^{-\ast}$ then $\len(\sigma) \leq \len(\phi)$, for all $\sigma \in \seq(\phi)$.
\end{enumerate}
\end{lemma}

Note, however, that each $\sigma \in \seq(\pi^\ast)$ is infinite.
%
%
%
%
%
%
%
\subsection{Semantics}
A \emph{\dlpa model} is a set $\Vals \subseteq \Props$ of propositional variables.
When $p \in \Vals$ then $p$ is true, and 
when $p \notin \Vals$ then $p$ is false.

The interpretation of an assignment $\alpha$ is in terms of a \emph{model update}.
The update of a model $\Vals$ by an assignment $\alpha$ is the new model $\Vals^\alpha$ such that:
\[
\Vals^\alpha =\set{\prop : \Vals \models \alpha(\prop)}
\]
where we suppose that when $\prop$ is not in the domain of $\alpha$ then $\alpha(\prop)$ equals $\prop$.
In particular, for the assignment ${+}\prop$ we have $\Vals^{{+}\prop} = \Vals \cup \set{\prop}$.
Given a sequence of assignments $\alpha_1 \ldots \alpha_n$, for the sake of readability, we sometimes write $\Vals^{\alpha_1 \cdots \alpha_n}$ instead of $( \cdots (\Vals^{\alpha_1})\cdots)^{\alpha_n}$.

Formulas $\phi$ are interpreted as \emph{sets of models} $\exto{\phi}$, while programs $\pi$ are interpreted by means of a (unique) \emph{relation between valuations} $\exto{\pi}$.
Just as in \pdl, the formal definition is by a mutual recursion.
It is given in Table~\ref{table:truthConds}.

\begin{table}
\begin{align*}
\begin{aligned}[t]
\exto{\prop} & = \set{\Vals : \prop \in \Vals}\\
\exto{\lnot\phi} & = 2^\Props \setminus \exto{\phi}\\
\exto{\phi \land \psi} & = \exto{\phi} \cap \exto{\psi}\\
\exto{\lbox{\pi}\phi} & = \set{\Vals : \text{if } \tuple{\Vals, \Vals'} \in \exto{\pi} \text{ then } \Vals' \in \exto{\phi}}\\
\end{aligned}
& \qquad
\begin{aligned}[t]
\exto{\alpha} & = \set{\tuple{\Vals,\Vals'} : \Vals' = \Vals^\alpha}\\
\exto{\pi_1; \pi_2} & = \exto{\pi_1} \circ \exto{\pi_2}\\
\exto{\pi_1 \ndet \pi_2} & = \exto{\pi_1} \cup \exto{\pi_2}\\
\exto{\pi^\ast} & = \bigcup_{n\in \Nats_0} ( \exto{\pi} )^n\\
\exto{\phi?} & = \set{\tuple{\Vals,\Vals} : \Vals \in \exto{\phi}}
\end{aligned}
\end{align*}
\caption{Interpretation of the \dlpa connectives}
\label{table:truthConds}
\end{table}

As usual, we also write $\Vals \models \phi$ to mean that $\Vals \in \exto{\phi}$.
Moreover, given a formula $\phi$, 
we say that $\phi$ is \emph{\dlpa valid} (noted $\models\phi$)
if and only if $\exto{\phi} = 2^\Props$,
and
we say that $\phi$ is \emph{\dlpa satisfiable}
if and only if $\exto{\phi} \neq \emptyset$.

For example, the formulas
$\lbox{{+}\prop}\top$,
$\lbox{{+}\prop}\phi \lequiv \lnot\lbox{{+}\prop}\lnot\phi$,
$\lbox{\pi}\top$,
$\lbox{{+}\prop}\prop$ and
$\lbox{{-}\prop}\lnot\prop$
are all \dlpa valid.
\subsection{Existing Proof Methods}
We now recall the existing methods for both model checking and satisfiability checking in $\dlpa$.
They either use a non-elementary reduction to propositional logic or a quadratic embedding into $\pdl$.
We then provide a linear reduction of satisfiability checking to model checking.
This justifies our focus on a tableaux method for model checking in the rest of the paper.
But first, let us recall some valid principles in $\dlpa$.
\begin{proposition}[\cite{HerzigLMT-Ijcai11}]
\label{pro:principles}
The following principles are valid in \dlpa:
\begin{enumerate}
\item
$\lbox{\alpha}\prop \lequiv \alpha(\prop)$
\item
$\lbox{\psi?}\phi \lequiv (\psi \limp \phi)$
\item
$\lbox{\pi}\lnot\phi \lequiv \lnot\lbox{\pi}\phi$
\item
$\lbox{\pi}(\phi \land \psi) \lequiv (\lbox{\pi}\phi \land \lbox{\pi}\psi)$
\item
$\lbox{\pi_1; \pi_2}\phi \lequiv \lbox{\pi_1}\lbox{\pi_2}\phi$
\item
$\lbox{\pi_1 \ndet \pi_2}\phi \lequiv (\lbox{\pi_1}\phi \land \lbox{\pi_2}\phi)$
\item
$\lbox{\pi^\ast}\phi \lequiv (\phi \land \lbox{\pi}\lbox{\pi^\ast}\phi)$
\item
From $\psi \limp (\phi \land \lbox{\pi^\ast}\psi)$ infer $\psi \limp \lbox{\pi^\ast}\phi$
\end{enumerate}
\end{proposition}

It follows from Proposition~\ref{pro:principles}.1--\ref{pro:principles}.6 plus the rule of substitution of valid equivalences that the star-free fragment of $\dlpa$ is reducible to propositional logic.
This however fails to provide an efficient theorem proving method because the reduced formula might be exponentially longer than the original formula.
In \cite{BalbianiHerzigTroquard-Lics13}, it is also shown that the Kleene star can be eliminated in $\dlpa$,
i.e., there is an algorithm that translates every formula in $\lang$ to an equivalent formula in $\lang^{-\ast}$.
Such translation, however, also leads to much longer formulas.
In fact, this is a non-elementary reduction because it starts from the innermost Kleene star operator.

Satisfiability checking in $\dlpa$ is shown to be in EXPTIME in \cite{BalbianiHerzigTroquard-Lics13}.
The proof is given via a translation to satisfiability checking in $\pdl$.
For every $\dlpa$ formula $\phi$, the translation $\tr$ returns a $\pdl$ formula which is obtained by just replacing each assignment $\pm\prop$ by an abstract $\pdl$ program $\act_{\pm\prop}$.
To guarantee that the abstract programs behave the same way as the original assignment, the following set of formulas $\Gamma_\phi$ is also used:
\[
\Gamma_\phi =
\begin{aligned}[t]
& \set{\lbox{\act_{{+}\prop}}\prop : \prop \in \Props_\phi}~\cup\\
& \set{\lbox{\act_{{-}\prop}}\lnot\prop : \prop \in \Props_\phi}~\cup\\
& \set{\ldia{\act_{\pm\prop}}\top : \pm\prop \in \Props_\phi}~\cup\\
& \set{\propb \limp \lbox{\act_{\pm\prop}}\propb : \prop, \propb \in \Props_\phi, \prop \neq \propb}~\cup\\
& \set{\lnot\propb \limp \lbox{\act_{\pm\prop}}\lnot\propb : \prop, \propb \in \Props_\phi, \prop \neq \propb}
& \end{aligned}
\]

\begin{proposition}[\cite{BalbianiHerzigTroquard-Lics13}]
Let $U_\phi$ be the $\pdl$ program $(\bigcup_{\prop \in \Props_\phi}(\act_{{+}\prop} \cup \act_{{-}\prop}))^\ast$.
For every $\dlpa$ formula $\phi$, $\phi$ is $\dlpa$ satisfiable if and only if
\[
\tr(\phi) \land \lbox{U_\phi}\left(\bigwedge \Gamma_\phi\right)
\]
is $\pdl$ satisfiable.
\end{proposition}

Note that, even though this reduction is polynomial, a quadratically longer formula is produced.
Precisely, the size of $\Gamma_\phi$ is bounded by $5 \len(\phi)^2$.
Moreover, if we consider the star-free fragment of $\dlpa$, this transformation is sub-optimal, because of the Kleene star operator in $U_\phi$.\footnotemark
\footnotetext{%
For the star-free fragment, a transformation without the Kleene star operator is also possible.
In this case, the program $\bigcup_{\prop \in \Props_\phi}(\act_{{+}\prop} \cup \act_{{-}\prop})$ must be iterated up to $\len(\phi)$, but this leads to an even longer formula.
}

If follows from the next result that satisfiability checking in $\dlpa$ can be linearly reduced to model checking in $\dlpa$.

\begin{proposition}
\label{pro:sc_to_mc}
Let a formula $\phi \in \lang$ be given.
Let $\Props_\phi = \set{\prop_1, \dots, \prop_n}$.
and
let $M_\phi$ be the $\dlpa$ program $({+}\prop_1 \ndet {-}\prop_1); \dots; ({+}\prop_n \ndet {-}\prop_n)$.
Formula $\phi$ is satisfiable
if and only if
$\Vals \models \ldia{M_\phi}\phi$ for any model $\Vals$.
\end{proposition}

\begin{proof}
It suffices to see that the interpretation of the program $M_\phi$ relates all possible valuations in the vocabulary of $\phi$, while leaving the other variables unchanged.
\qed
\end{proof}

The operation $\lbox{M_\phi}$ works as a master modality and thus $\ldia{M_\phi}$ works as its dual.
Because it does not contain the Kleene star operator, the length of $\ldia{M_\phi}\phi$ is bounded by $3\len(\phi)$.
Also note that, in particular, $\phi$ is satisfiable if and only if $\Vals = \emptyset$ satisfies $\ldia{M_\phi}\phi$.
This means that the input $(\Vals, \ldia{M_\phi}\phi)$ for the model checking problem is also linear on the length of $\phi$.
Therefore, in order to perform satisfiability checking in $\dlpa$, one could take advantage of an efficient algorithm for model checking in $\dlpa$.
This motivates the tableaux methods presented in the next section.

Before concluding this section, let us recall that, in $\dlpa$, model checking has the same computational complexity as satisfiability checking.
This follows from Proposition~\ref{pro:sc_to_mc} above and Proposition~\ref{pro:mc_to_sc} below.

\begin{proposition}[\cite{BalbianiHerzigTroquard-Lics13}]
\label{pro:mc_to_sc}
For every valuation $\Vals$ and formula $\phi$, $\Vals \in \exto{\phi}$ if and only if the formula
\[
\phi \land \left( \bigwedge_{\prop \in \Props_\phi \cap \Vals}\prop \right) \land \left( \bigwedge_{\prop \in \Props_\phi \setminus \Vals}\lnot\prop \right)
\]
is $\dlpa$ satisfiable.
\end{proposition}
%
\section{A Tableaux Method for Star-Free \dlpa}
\label{sec:tableaux_star_free}
%
In this section, we define a model checking procedure for the star-free fragment of $\dlpa$ using analytic tableaux.
We start with some useful definitions. 

A \emph{labeled formula} is a pair $\lambda = \tuple{\sigma, \phi}$, where $\sigma = \alpha_1 \dots \alpha_n$ is a (possibly empty) sequence of propositional assignments and $\phi \in \lang$.
A \emph{branch} is a set of labelled formulas.

\begin{definition}[Tableau]
\label{def:tableau}
Let $\Vals \subseteq \Props$ and $\phi_0 \in \lang^{-\ast}$.
The initial branch for $(\Vals, \phi_0)$ is the the set
$
\branch_0 =
\set{\tuple{\emptyseq, \prop} : \prop \in \Props_{\phi_0} \cap \Vals}
\cup
\set{\tuple{\emptyseq, \lnot\prop} : \prop \in \Props_{\phi_0} \setminus \Vals}
\cup
\set{\tuple{\emptyseq, \phi_0}}
$.
A tableau for $(\Vals, \phi_0)$ is a set of branches $\Tableau$ that satisfies one of the following two conditions:
\begin{enumerate}
\item
$\Tableau = \set{\branch_0}$,
which is called the \emph{initial tableau} for $(\Vals, \phi_0)$.
\item
$\Tableau = (\Tableau' \setminus \set{\branch}) \cup \Branches$, where 
$\Tableau'$ is a tableau for $(\Vals, \phi_0)$ containing the branch $\branch$ and 
$\Branches$ is a set of $k$ branches $\set{\branch \cup \branch_1, \dots, \branch \cup \branch_k}$ generated by one of the following \emph{tableau rules} below:\onlylong{\footnotemark
\footnotetext{Some of these rules are also presented in the more traditional numerator-denominator form in Table~\ref{tab:rules} of the Appendix.}
}
\begin{description}
\item[(R$\lnot$)]
$\tuple{\sigma, \lnot\lnot\phi} \in \branch$
implies
$k = 1$ and
$\branch_1 = \set{\tuple{\sigma, \phi}}$.
\item[(R$\land$)]
$\tuple{\sigma, \phi \land \psi} \in \branch$
implies
$k = 1$ and
$\branch_1 = \set{\tuple{\sigma, \phi}, \tuple{\sigma, \psi}}$.
\item[(R$\lor$)]
$\tuple{\sigma, \lnot(\phi \land \psi)} \in \branch$
implies
$k = 2$,
$\branch_1 = \set{\tuple{\sigma, \lnot\phi}}$
and $\branch_2 = \set{\tuple{\sigma, \lnot\psi}}$.
\item[(R$\lbox{\alpha}$)]
$\tuple{\sigma, \lbox{\alpha}\phi} \in \branch$
implies
$k = 1$ and
\\~\hfill
$\branch_1 = \set{ \tuple{\sigma\alpha, \phi} } \cup 
\set{ \tuple{\sigma\alpha, \prop} : \alpha(\prop) = \top } \cup
\set{ \tuple{\sigma\alpha, \lnot\prop} : \alpha(\prop) = \bot } $.
\item[(R$\ldia{\alpha}$)]
$\tuple{\sigma, \lnot\lbox{\alpha}\phi} \in \branch$
implies
$k = 1$ and
\\~\hfill
$\branch_1 = \set{ \tuple{\sigma\alpha, \lnot\phi} } \cup 
\set{ \tuple{\sigma\alpha, \prop} : \alpha(\prop) = \top } \cup
\set{ \tuple{\sigma\alpha, \lnot\prop} : \alpha(\prop) = \bot } $.
\item[(R$\lbox{?}$)]
$\tuple{\sigma, \lbox{\psi?}\phi} \in \branch$
implies
$k = 2$,
$\branch_1 = \set{\tuple{\sigma, \lnot\psi}}$
and $\branch_2 = \set{\tuple{\sigma, \phi}}$.
\item[(R$\ldia{?}$)]
$\tuple{\sigma, \lnot\lbox{\psi?}\phi} \in \branch$
implies
$k = 1$ and
$\branch_1 = \set{\tuple{\sigma, \psi}, \tuple{\sigma, \lnot\phi}}$.
\item[(R$\lbox{{;}}$)]
$\tuple{\sigma, \lbox{\pi_1; \pi_2}\phi} \in \branch$
implies
$k = 1$ and
$\branch_1 = \set{\tuple{\sigma, \lbox{\pi_1}\lbox{\pi_2}\phi}}$.
\item[(R$\ldia{{;}}$)]
$\tuple{\sigma, \lnot\lbox{\pi_1; \pi_2}\phi} \in \branch$
implies
$k = 1$ and
$\branch_1 = \set{\tuple{\sigma, \lnot\lbox{\pi_1}\lbox{\pi_2}\phi}}$.
\item[(R$\lbox{\ndet}$)]
$\tuple{\sigma, \lbox{\pi_1 \ndet \pi_2}\phi}$
implies
$k = 1$ and
$\branch_1 = \set{\tuple{\sigma, \lbox{\pi_1}\phi}, \tuple{\sigma, \lbox{\pi_2}\phi}}$.
\item[(R$\ldia{\ndet}$)]
$\tuple{\sigma, \lnot\lbox{\pi_1 \ndet \pi_2}\phi} \in \branch$
implies
$k = 2$,
$\branch_1 = \set{\tuple{\sigma, \lnot\lbox{\pi_1}\phi}}$
and $\branch_2 = \set{\tuple{\sigma, \lnot\lbox{\pi_2}\phi}}$.
%
%
\item[(RP1)]
$\set{\tuple{\sigma, \prop}, \tuple{\sigma\alpha, \psi}} \subseteq \branch$
for some $\psi$ and $\prop \not\in \dom(\alpha)$
implies
$k = 1$ and\\
$~$\hfill
$\branch_1 = \set{\tuple{\sigma\alpha, \prop}}$.
\item[(RP2)]
$\set{\tuple{\sigma, \lnot\prop}, \tuple{\sigma\alpha, \psi}} \subseteq \branch$
for some $\psi$ and $\prop \not\in \dom(\alpha)$
implies
$k = 1$ and\\
$~$\hfill
$\branch_1 = \set{\tuple{\sigma\alpha, \lnot\prop}}$.
\end{description}
\end{enumerate}
\end{definition}

The initial tableau corresponds to the input of the problem in the tableau.
Rules R$\lnot$, R$\land$ and R$\lor$ are the standard tableaux rules for Boolean connectives.
RP1 and RP2 (propagation rules) propagate literals whose the truth value is not changed by assignments:
if the model updated by $\sigma$ satisfies $\prop$ and $\alpha$ does not change the truth value of $\prop$ then the model updated by $\sigma\alpha$ also satisfies $\prop$.
The other rules just reflect the semantic definition of the corresponding programs.
For instance, for the rule R$\lbox{\alpha}$, if the model updated by the sequence of assignments $\sigma$ satisfies $\lbox{\alpha}\phi$ then the model updated by the sequence $\sigma\alpha$ satisfies $\phi$.
Note that they also correspond to the validities 1--6 in Proposition~\ref{pro:principles}.

A branch $\branch$ is \emph{blatantly inconsistent} if and only if $\branch$ contains both
$\tuple{\sigma, \phi}$ and $\tuple{\sigma, \lnot\phi}$,
for some $\sigma$ and $\phi$.
A branch $\branch$ is \emph{closed} if and only if it is blatantly inconsistent.
A tableau is closed if and only if all its branches are closed.
A tableau is \emph{open} if and only if it is not closed. 

The idea is that, if there is a closed tableau for the input $(\Vals, \phi_0)$ then $\Vals \not\models \phi_0$.
On the other hand, if there is no closed tableau for $(\Vals, \phi_0)$ then $\Vals \models \phi_0$.

\begin{example}
Table~\ref{tab:example1} shows how the method can be used to prove that the model $\Vals = \set{\prop, \propb}$ does not satisfy the formula $\phi_0 = \lnot\lbox{{+}\prop \ndet {-}\prop}\propb$.
In the table, lines~1--3 consist of the initial tableau for the input $(\Vals, \phi_0)$.
Rule applications are indicated between parentheses on the left of each line.
Line~4 is generated by the application of R$\lbox{\ndet}$ to line~3.
This generates two different branches.
The rule applications continue until both branches are closed.
\end{example}

\begin{table}
\[
\begin{array}{c}
\begin{array}{l@{\quad}l@{\quad}l@{\quad}l}
1. & \emptyseq & \prop\\
2. & \emptyseq & \propb\\
3. & \emptyseq &  \lnot\lbox{{+}\prop \ndet {-}\prop}\propb\\
\end{array}\\
\hline
\begin{array}{c|c}
\begin{array}[t]{l@{\quad}l@{\quad}l@{\quad}l}
4. & \emptyseq & \lnot\lbox{{+}\prop}\propb & \text{(R$\ldia{\ndet}$: 3)}\\
5. & {+}\prop & \prop & \text{(R$\ldia{\alpha}$: 4)}\\
6. & {+}\prop & \lnot\propb & \text{(R$\ldia{\alpha}$: 4)}\\
7. & {+}\prop & \propb & \text{(RP1: 2, 5)}\\
8. && \text{(closed)} & \text{(6, 7)}
\end{array}
& 
\begin{array}[t]{l@{\quad}l@{\quad}l@{\quad}l}
4. & \emptyseq & \lnot\lbox{{-}\prop}\propb & \text{(R$\ldia{\ndet}$: 3)}\\
5. & {-}\prop & \lnot\prop & \text{(R$\ldia{\alpha}$: 4)}\\
6. & {-}\prop & \lnot\propb & \text{(R$\ldia{\alpha}$: 4)}\\
7. & {+}\prop & \propb & \text{(RP1: 2, 5)}\\
8. && \text{(closed)} & \text{(6, 7)}
\end{array}
\end{array}
\end{array}
\]
\caption{%
\label{tab:example1}
Tableau for $\Vals = \set{\prop, \propb}$ and $\phi_0 = \lnot\lbox{{+}\prop \ndet {-}\prop}\propb$}
\end{table}

\begin{example}
Table~\ref{tab:example_2} shows how the method can be used to prove that the model $\Vals = \emptyset$ satisfy the formula
$\phi_0 = \lbox{\lnot\prop?; {+}\prop}\prop$.
Note that RP2 is not applicable to the labelled formulas in lines~1 and 5 because $\prop \in \dom({+}\prop)$.
Thus, the branch on the right remains open, which means that $\Vals \models \phi_0$.
\end{example}

\begin{table}
\[
\begin{array}{c}
\begin{array}{l@{\quad}l@{\quad}l@{\quad}l}
1. & \emptyseq & \lnot\prop\\
2. & \emptyseq & \lbox{\lnot\prop?; {+}\prop}\prop\\
3. & \emptyseq & \lbox{\lnot\prop?}\lbox{{+}\prop}\prop & \text{(R$\lbox{{;}}$: 2)}\\
\end{array}\\
\hline
\begin{array}[t]{c|c}
\begin{array}[t]{l@{\quad}l@{\quad}l@{\quad}l}
4. & \emptyseq & \lnot\lnot\prop & \text{(R$\lbox{?}$: 3)}\\
5. & \emptyseq & \prop & \text{(R$\lnot$: 4)}\\
&& \text{(closed)} & \text{(1, 5)}
\end{array}
& 
\begin{array}[t]{c}
\begin{array}[t]{l@{\quad}l@{\quad}l@{\quad}l}
4. & \emptyseq & \lbox{{+}\prop}\prop & \text{(R$\lbox{?}$: 3)}\\
5. & {+}\prop & \prop & \text{(R$\lbox{\alpha}$: 4)}\\
&& \text{(open)}
\end{array}\\
\end{array}
\end{array}
\end{array}
\]
\caption{%
\label{tab:example_2}
Tableau for $\Vals = \emptyset$ and $\phi_0 = \lbox{\lnot\prop?; {+}\prop}\prop$
}
\end{table}

In the sequel, we show the soundness of the method.
The idea is to show that, if $\Vals \models \phi_0$, then successive rule applications can never close the tableau.
But first, a useful definition and a lemma are presented.

\begin{definition}[Consistent Branch]
A branch $\branch$ is consistent if and only if
$\Vals^\sigma \models \phi$ for every $\tuple{\sigma, \phi} \in \branch$.
\end{definition}

\begin{lemma}[Consistency Preservation]
\label{lem:sat_preservation_1}
For each tableau rule $\rho$, if branch $\branch$ is consistent, then the set of branches $\Branches$ generated by the application of $\rho$ to $\branch$ contains a consistent branch.
\end{lemma}

\begin{theorem}[Soundness]
\label{theo:soundness_1}
If $\Vals \models \phi_0$ then there is no closed tableau for $(\Vals, \phi_0)$.
\end{theorem}

We now address the completeness of the method.
The idea is to show that, if the tableau remains open after all possible applications of the tableau rules, then $\Vals 
\models \phi_0$.
But first, some useful definitions are presented.

\begin{definition}[Witness]
A witness to rule $\rho$ in branch $\branch$ is a labelled formula $\tuple{\sigma, \phi} \in \branch$ allowing the application of $\rho$.
\end{definition}

For example, $\tuple{\emptyseq, \lnot\lnot\prop}$ is a witness to R$\lnot$, and 
$\tuple{\beta, \lnot\lbox{{+}\prop, {-}\propb}\prop}$ is a witness to R$\ldia{\alpha}$.
Moreover, the formula $\tuple{\sigma, \prop}$ is a witness to RP1 in $\branch$ if there is a formula $\tuple{\sigma\alpha, \psi} \in \branch$ and $\prop \not\in \dom(\alpha)$.

\begin{definition}[Saturated Tableau]
The label $\sigma$ in the branch $\branch$ is saturated under the tableau rule $\rho$ if and only if for each witness $\tuple{\sigma, \phi}$ to $\rho$ in $\branch$, $\branch$ contains some $\branch_i$ generated by the application of $\rho$ to $\branch$.
The branch $\branch$ is saturated under the tableau rule $\rho$ if and only if all its labels are saturated.
A branch is saturated if and only if it is saturated under all tableau rules.
A tableau is saturated if and only if all its branches are saturated.
\end{definition}

\begin{theorem}[Completeness]
\label{theo:completeness_1}
If there is no closed tableau for $(\Vals, \phi_0)$ then $\Vals \models \phi_0$.
\end{theorem}
\section{An Optimal Procedure for Star-Free \dlpa}
\label{sec:optimal_star_free}
\newcommand{\Satsf}{mcTableau}
\begin{table}
\begin{tabular}{r@{:~~}l}
 1 & \Input{$(\Vals, \phi_0)$}\\
 2 & \Output{%
$\begin{cases}
\text{\True}, & \text{if $\branch$ is satisfiable}\\
\text{\False}, & \text{otherwise}
\end{cases}$
}\\
 3 & \Begin\\
 4 & \quad \Satsf$(\branch_0)$\\
 5 & \End
\bigskip\\
 6 & \Function{\Satsf}$(\branch)$\\
 7 & \Begin\\
 8 & \quad \If{$\branch$ contains an applicable witness $\lambda$ to a rule $\rho \in \set{\text{R}\lnot, \text{R}\land, \text{R}\ldia{?}, \text{R}\lbox{{;}}, \text{R}\ldia{{;}}, \text{R}\lbox{\ndet}}$}\\
 9 & \quad\quad $\branch_1 \leftarrow$ the branch generated by the application of $\rho$ to $\branch$ using $\lambda$ as witness\\
10 & \quad\quad mark $\lambda$ as `non-applicable'\\
11 & \quad\quad \Return{\Satsf$(\branch \cup \branch_1)$}\\
12 & \quad \ElseIf{$\branch$ contains an applicable witness $\lambda$ to a rule $\rho \in \set{\text{R}\lor, \text{R}\lbox{?}, \text{R}\ldia{\ndet}, \text{RC}}$}\\
13 & \quad\quad $\Branches \leftarrow$ the set of branches $\set{\branch_1, \dots, \branch_n}$ generated by the application of $\rho$ to $\branch$ using $\lambda$ as witness\\
14 & \quad\quad mark $\lambda$ as `non-applicable'\\
15 & \quad\quad \ForEach{$\branch_i \in \Branches$}\\
16 & \quad\quad\quad \If{\Satsf$(\branch \cup \branch_i)$ = \True}\\
17 & \quad\quad\quad\quad \Return{\True}\\
18 & \quad\quad\quad \EndIf\\
19 & \quad\quad \EndFor\\
20 & \quad\quad \Return{\False}\\
21 & \quad \EndIf\\
22 & \quad \textbf{while} \parbox[t]{9cm}{%
		there is an atomic program $\alpha$ such that $\branch$ contains an applicable witness $\lambda = \tuple{\sigma, \phi}$
		to rule $\rho \in \set{\text{R}\lbox{\alpha}, \text{R}\ldia{\alpha}}$,
		where $\phi = \lbox{\alpha}\psi$ or $\phi = \ldia{\alpha}\psi$
		\textbf{do}
		}\\
23 & \quad\quad $\branch_1 \leftarrow$ the branch generated by the application of $\rho$ to $\branch$ using $\lambda$ as witness\\
24 & \quad\quad mark $\lambda$ as `non-applicable'\\
25 & \quad\quad \textbf{while} \parbox[t]{9cm}{%
					$\branch$ contains an applicable witness $\lambda' = \tuple{\sigma, \phi'}$
					to rule $\rho \in \set{\text{R}\lbox{\alpha}, \text{R}\ldia{\alpha}}$,
					where $\phi' = \lbox{\alpha}\psi'$ or $\phi' = \ldia{\alpha}\psi'$
					\textbf{do}
					}\\
26 & \quad\quad\quad\quad $\branch'_1 \leftarrow$ the branch generated by the application of $\rho$ to $\branch$ using $\lambda'$ as witness\\
27 & \quad\quad\quad\quad mark $\lambda'$ as `non-applicable'\\
28 & \quad\quad\quad\quad $\branch_1 \leftarrow \branch_1 \cup \branch'_1$\\
29 & \quad\quad \EndWhile\\
30 & \quad\quad \While{$\branch \cup \branch_1$ contains an applicable witness $\lambda''$ to rule $\rho \in \set{\text{RP}_1, \text{RP}_2}}$\\
31 & \quad\quad\quad $\branch'_1 \leftarrow$ the branch generated by the application of $\rho$ to $\branch$ using $\lambda''$ as witness\\
32 & \quad\quad\quad mark $\lambda''$ as `non-applicable'\\
33 & \quad\quad\quad $\branch_1 \leftarrow \branch_1 \cup \branch'_1$\\
34 & \quad\quad \EndWhile\\
35 & \quad\quad \If{\Satsf$(\branch_1)$ = \False}\\
36 & \quad\quad\quad \Return{\False}\\
37 & \quad\quad \EndIf\\
38 & \quad \EndWhile\\
39 & \quad \Return{\True}\\
40 & \End\\
\end{tabular}
\caption{%
\label{tab:alg1}%
Algorithm implementing the tableaux method for star-free $\dlpa$
}
\end{table}

In this section we define an algorithm to check whether $\Vals \models \phi_0$, for $\phi_0 \in \lang^{-\ast}$.
Such an algorithm is displayed in Table~\ref{tab:alg1}.
It implements the tableaux method using the recursive function \Satsf.
It takes as argument a tableau branch $\branch$ and returns whether $\branch$ is consistent.
When called with the initial tableau for $(\Vals, \phi_0)$ it returns whether $\Vals \models \phi_0$.
The execution of \Satsf\ explores in a depth-first manner a tree whose nodes are tableau branches and each child is generated by the application of a tableau rule to its parent.

The rules are applied in a specific order and, after the application of a rule, the witness is marked `non-applicable', thus avoiding an infinite loop.
Lines 8--21 perform what is called `local saturation'.
That is, only rules that do not create labelled formulas with different labels than that of the witness are applied.
Its first part (lines 8--11) applies rules that do not create more than one branch in the tableau.
Its second part (lines 12--21) apples rules that create more than one branch in the tableau.
At the end of the local saturation, only witnesses to rules R$\lbox{\alpha}$, R$\ldia{\alpha}$ remain.
Note that no new label is created in the local saturation part, which means that there can be no witnesses to rules RP1 and RP2.
Then, in lines 22--38 the algorithm performs what is called `successor creation'.
First (line 22), it tests whether there is a successor to be created, i.e., if there is a witness $\lambda$ to R$\lbox{\alpha}$ or R$\ldia{\alpha}$.
It creates the successor (line 23) and then marks the witness as `non-applicable' (line 24).
After that (lines 25--34), it propagates the suitable formulas to the successor, as follows:
assume that the witness is $\lambda = \tuple{\sigma, \lbox{\alpha}\psi}$.
Then, for every labelled formula $\tuple{\sigma, \lbox{\alpha}\psi'}$ and $\tuple{\sigma, \ldia{\alpha}\psi'}$ there must be a labelled formula $\tuple{\sigma\alpha, \psi'}$ in the successor.
This is done in lines 25--29.
And also, every labelled formula $\tuple{\sigma, \prop}$ (resp. $\tuple{\sigma, \lnot\prop}$) must be propagated,
i.e., there must be a labelled formula $\tuple{\sigma\alpha, \prop}$ (resp. $\tuple{\sigma\alpha, \lnot\prop}$) in the successor $\branch_1$.
This is done in lines 30--34.
The last part (lines 35--37) makes a recursive call to \Satsf\ with the $\branch_1$.
The current branch is considered satisfiable if all recursive calls return \True.

This algorithm has two important features.
First, its successor creation part guarantees that each time \Satsf\ is called with branch $\branch$ as argument, all the labelled formulas in $\branch$ have the same label.
Second, the first feature implies that the list of successors created during successive recursive calls of \Satsf\ corresponds to one execution trace from input formula $\phi_0$.
These are the key arguments used in the proof of complexity result below.

\begin{theorem}[Termination]
\label{theo:termination_1}
The algorithm in Table~\ref{tab:alg1} halts for every input $(\Vals, \phi_0)$.
\end{theorem}

\begin{theorem}[Complexity]
\label{theo:complexity_1}
The amount of memory used by the algorithm in Table~\ref{tab:alg1} is a polynomial function of the length of the input $(\Vals, \phi_0)$.
\end{theorem}

Therefore, the algorithm in Table~\ref{tab:alg1} works in space polynomial in the length of the input.
This is an optimal algorithm, given that the satisfiability problem in star-free \dlpa is PSPACE-complete \cite{HerzigLMT-Ijcai11}.
%
\section{A Tableaux Method for Full \dlpa}
\label{sec:tableaux}
%
In this section, we define an extension of the tableaux method that also takes into account the Kleene star operator.

\begin{definition}[Tableau]
Let $(\Vals, \phi_0)$ be the input under concern
(thus, the initial tableau is the same as in Definition~\ref{def:tableau}).
The tableau rules for full $\dlpa$ are those of Definition~\ref{def:tableau} 
plus the following ones:
\begin{description}
\item[(R$\lbox{\ast}$)]
$\tuple{\sigma, \lbox{\pi^\ast}\phi} \in \branch$
implies
$k = 1$ and
$\branch_1 = \set{\tuple{\sigma, \phi}, \tuple{\sigma, \lbox{\pi}\lbox{\pi^\ast}\phi}}$.
\item[(R$\ldia{\ast}$)]
$\tuple{\sigma, \lnot\lbox{\pi^\ast}\phi} \in \branch$
implies
$k = 2$,
$\branch_1 = \set{\tuple{\sigma, \lnot\phi}}$
and $\branch_2 = \set{\tuple{\sigma, \phi}, \tuple{\sigma, \lnot\lbox{\pi}\lbox{\pi^\ast}\phi}}$.
\end{description}
\end{definition}

The two rules above reflect the fix point property of Proposition~\ref{pro:principles}.7.
For instance, if the model $\Vals^\sigma \models \lbox{\pi^\ast}\phi$ then $\Vals^\sigma \models \phi$ and also $\Vals^\sigma \models \lbox{\pi}\lbox{\pi^\ast}\phi$.

\begin{definition}[Fulfillment]
An eventuality $\tuple{\sigma, \lnot\lbox{\pi^\ast}\phi}$) is fulfilled in a tableau branch $\branch$ if and only if
there is a (possibly empty) execution trace $\sigma' \in \seq(\pi)$ such that $\tuple{\sigma\sigma', \lnot\phi} \in \branch$.
\end{definition}

\begin{definition}[Closed Branch]
A branch $\branch$ is closed if and only if
(1) $\branch$ is blatantly inconsistent or 
(2) $\branch$ is saturated and contains an unfulfilled eventuality.
\end{definition}

\begin{example}
\label{eg:star}
Table~\ref{tab:example_star} shows how the method can be used to prove that model $\Vals = \set{\prop, \propb}$ does not satisfy the formula $\phi_0 = \lnot\lbox{({+}\prop \ndet {-}\prop)^\ast}\propb$.
The leftmost branch is closed because it is blatantly inconsistent.
In the branch of the middle, the same pattern will be repeated indefinitely.
Thus, it is an infinite brunch, but it is saturated.
Since the eventuality in line 3 is not fulfilled, it is also closed.
The right-most branch is analogous to the one in the middle.
\end{example}

If the input formula contains a sub-formula of the form $\lnot\lbox{\pi^\ast}\phi$, the method invariably creates tableaux with infinite branches that repeat the same pattern over and over again, as in Example~\ref{eg:star}.
The repetition can be detected and it is possible to provide a terminating algorithm.
This is presented in Section~\ref{sec:optimal}.
Here, we show the correctness of the method presented so far.

\begin{table}
\[
\begin{array}[t]{c}
\begin{array}[t]{l@{\quad}l@{\quad}l@{\quad}l}
1. & \emptyseq & \prop\\
2. & \emptyseq & \propb\\
3. & \emptyseq & \lnot\lbox{({+}\prop \ndet {-}\prop)^\ast}\propb\\
\end{array}\\
\hline
\begin{array}[t]{c|c}
\begin{array}[t]{c}
\begin{array}[t]{l@{\quad}l@{\quad}l@{\quad}l}
4. & \emptyseq & \lnot\propb & \text{(R$\ldia{\ast}$: 3)}\\
5. & & \text{(blat. inc.)} & \text{(2, 4)}\\
\end{array}
\end{array}
&
\begin{array}[t]{c}
\begin{array}[t]{l@{\quad}l@{\quad}l@{\quad}l}
4. & \emptyseq & \lnot\lbox{{+}\prop \ndet {-}\prop}\lbox{({+}\prop \ndet {-}\prop)^\ast}\propb & \text{(R$\ldia{\ast}$: 3)}\\
\end{array}\\
\hline
\begin{array}[t]{c|c}
\begin{array}[t]{c}
\begin{array}[t]{l@{\quad}l@{\quad}l@{\quad}l}
5. & \emptyseq & \lnot\lbox{{+}\prop}\lbox{({+}\prop \ndet {-}\prop)^\ast}\propb & \text{(R$\ldia{\ndet}$: 3)}\\
6. & {+}\prop & \prop & \text{(R$\ldia{\alpha}$: 5)}\\
7. & {+}\prop & \lnot\lbox{({+}\prop \ndet {-}\prop)^\ast}\propb & \text{(R$\ldia{\alpha}$: 5)}\\
8. & {+}\prop & \propb & \text{(RP1: 2, 6)}\\
&& \quad\vdots\\
&& \text{(closed)}
\end{array}
\end{array}
&
\begin{array}[t]{c}
\begin{array}[t]{c@{\quad}l@{\quad}l@{\quad}l}
\vdots\\
\text{(closed)}\\
\end{array}
\end{array}
\end{array}
\end{array}
\end{array}
\end{array}
\]
\caption{%
\label{tab:example_star}
Tableau for $\Vals = \set{\prop, \propb}$ and $\phi_0 = \lnot\lbox{({+}\prop \ndet {-}\prop)^\ast}\propb$
}
\end{table}

\begin{lemma}[Consistency Preservation]
\label{lem:sat_preservation_2}
For each tableau rule $\rho$, if branch $\branch$ is consistent, then the set of branches $\Branches$ generated by the application of $\rho$ to $\branch$ contains a consistent branch.
\end{lemma}

\begin{theorem}[Soundness]
\label{theo:soundness_2}
If $\Vals \models \phi_0$ then there is no closed tableau for $(\Vals, \phi_0)$.
\end{theorem}

\begin{theorem}[Completeness]
\label{theo:completeness_2}
If there is no closed tableau for $(\Vals, \phi_0)$ then $\Vals \models \phi_0$.
\end{theorem}
\section{An EXPTIME Procedure for Full \dlpa}
\label{sec:optimal}
In this section, we define a procedure to model check formulas in $\lang$.
As before, we define an algorithm.
Here, it must detect the aforementioned repetitions of the applications of R$\ldia{\ast}$ in the tableau.
This is done by performing equality tests.
A label $\sigma_1$ is said to be equal to a label $\sigma_2$ if and only if the set of formulas labelled by $\sigma_1$ and $\sigma_2$ are the same.
More formally we have:

\begin{definition}[Equality]
Let $\sigma_1$ and $\sigma_2$ be two labels in the tableau $\Tableau$.
Label $\sigma_1$ is equal to label $\sigma_2$ (noted $\sigma_1 = \sigma_2$)
if and only if
there are two branches $\branch_1, \branch_2 \in \Tableau$ such that $\set{\phi : \tuple{\sigma_1, \phi} \in \branch_1} = \set{\phi : \tuple{\sigma_2, \phi} \in \branch_2}$.
\end{definition}

An equality test between labels can prevent the tableau to enter in an infinite loop.
Then one can try to provide an algorithm that is similar to the one in Section~\ref{sec:optimal_star_free}, by first adding rules R$\lbox{\ast}$ and R$\ldia{\ast}$ in their suitable places and the equality test just before the exploration of a new successor.
Such an algorithm works, but is not optimal.
For instance, the application of the method to the formula $\lbox{({+}\prop_1 \ndet {-}\prop_1 \ndet \cdots \ndet {+}\prop_n \ndet {-}\prop_n)^\ast}\prop$ creates $2^n$ different successors from a single tableau branch.
Then such a method may explore a tree whose the number of nodes is bounded by $2^{2^{\len(\phi_0)}}$.
However, satisfiability in $\dlpa$ is proven to be in EXPTIME.

A different technique than that in Section~\ref{sec:optimal_star_free} must be employed in order to obtain a more efficient method for full $\dlpa$.
Such a technique is implemented in the algorithm of Table~\ref{tab:alg2}.
It is somewhat similar to the algorithm in Section~\ref{sec:optimal_star_free}, but there are some important differences.
The most important ones are the addition of the equality test in lines 17--19 and the fact that this algorithm now maintains the entire tableau $\Tableau$ in memory.
It does not uses a recursive function anymore, for it now uses the tableau $\Tableau$ as the search tree.
Once the initial tableau for $(\Vals, \phi_0)$ is created in line 4, it enters a loop that finishes when $\Tableau$ is closed or saturated (recall that a branch is also considered to be closed if it is saturated and contains an unfulfilled eventuality).
As before, there is a `local saturation' part (lines 9--16) and a `successor creation' part (lines 20--33).
In lines 36--40, the algorithm tests whether $\Tableau$ is still open to return the right answer.

\begin{table}
\begin{tabular}{r@{:~~}l}
 1 & \Input{$(\Vals, \phi_0)$}\\
 2 & \Output{%
$\begin{cases}
\text{\True}, & \text{if $\phi_0$ is satisfiable}\\
\text{\False}, & \text{otherwise}
\end{cases}$
}\\
 3 & \Begin\\
 4 & \quad $\Tableau \leftarrow \set{\branch_0}$\\
 5 & \quad \While{$\Tableau$ is open and unsaturated}\\
 6 & \quad\quad pick an open and unsaturated branch $\branch \in \Tableau$\\
 7 & \quad\quad \While{$\branch$ is open and unsaturated}\\
 8 & \quad\quad\quad pick an open unsaturated label $\sigma$ of $\branch$\\
 9 & \quad\quad\quad \If{$\lambda = \tuple{\sigma, \phi} \in \branch$ is an applicable witness to a rule $\rho \in \set{\text{R}\lnot, \text{R}\land, \text{R}\ldia{?}, \text{R}\lbox{{;}}, \text{R}\ldia{{;}}, \text{R}\lbox{\ndet}, \text{R}\lbox{\ast}}$}\\
10 & \quad\quad\quad\quad $\branch_1 \leftarrow$ the branch generated by the application of $\rho$ to $\branch$ using $\lambda$ as witness\\
11 & \quad\quad\quad\quad mark $\lambda$ as `non-applicable'\\
12 & \quad\quad\quad\quad $\Tableau \leftarrow (\Tableau \setminus \set{\branch}) \cup \set{\branch \cup \branch_1}$\\
13 & \quad\quad\quad \ElseIf{$\lambda = \tuple{\sigma, \phi} \in \branch$ is an applicable witness to a rule $\rho \in \set{\text{R}\lor, \text{R}\lbox{?}, \text{R}\ldia{\ndet}, \text{R}\ldia{\ast}, \text{RC}}$}\\
14 & \quad\quad\quad\quad $\set{\branch_1, \branch_2} \leftarrow$ the branches generated by the application of $\rho$ to $\branch$ using $\lambda$ as witness\\
15 & \quad\quad\quad\quad mark $\lambda$ as `non-applicable'\\
16 & \quad\quad\quad\quad $\Tableau \leftarrow (\Tableau \setminus \set{\branch}) \cup \set{\branch \cup \branch_1, \branch \cup \branch_2}$\\
17 & \quad\quad\quad \ElseIf{there is a label $\sigma'$ in $\Tableau$ such that $\sigma = \sigma'$}\\
18 & \quad\quad\quad\quad mark all formulas in $\branch$ labelled by $\sigma$ as `non-applicable'\\
19 & \quad\quad\quad\quad \If{$\sigma'$ is closed} close branch $\branch$ \EndIf\\
20 & \quad\quad\quad \textbf{else if} \parbox[t]{9cm}{%
					there is an atomic program $\alpha$
					such that $\branch$ contains an applicable witness $\lambda = \tuple{\sigma, \phi}$
					to rule $\rho \in \set{\text{R}\lbox{\alpha}, \text{R}\ldia{\alpha}}$,
					where $\phi = \lbox{\alpha}\psi$ or $\phi = \ldia{\alpha}\psi$
					\textbf{do}
					}\\
21 & \quad\quad\quad\quad $\branch_1 \leftarrow$ the branch generated by the application of $\rho$ to $\branch$ using $\lambda$ as witness\\
22 & \quad\quad\quad\quad mark $\lambda$ as `non-applicable'\\
23 & \quad\quad\quad\quad \textbf{while} \parbox[t]{9cm}{%
						$\branch$ contains an applicable witness $\lambda' = \tuple{\sigma, \phi'}$
						to rule $\rho \in \set{\text{R}\lbox{\alpha}, \text{R}\ldia{\alpha}}$,
						where $\phi' = \lbox{\alpha}\psi'$ or $\phi' = \ldia{\alpha}\psi'$
						\textbf{do}
						}\\
24 & \quad\quad\quad\quad $\branch'_1 \leftarrow$ the branch generated by the application of $\rho$ to $\branch$ using $\lambda'$ as witness\\
25 & \quad\quad\quad\quad mark $\lambda'$ as `non-applicable'\\
26 & \quad\quad\quad\quad $\branch_1 \leftarrow \branch_1 \cup \branch'_1$\\
27 & \quad\quad\quad\quad \EndWhile\\
28 & \quad\quad\quad\quad \While{$\branch \cup \branch_1$ contains an applicable witness $\lambda''$ to rule $\rho \in \set{\text{RP}_1, \text{RP}_2}}$\\
29 & \quad\quad\quad\quad\quad $\branch'_1 \leftarrow$ the branch generated by the application of $\rho$ to $\branch$ using $\lambda''$ as witness\\
30 & \quad\quad\quad\quad\quad mark $\lambda''$ as `non-applicable'\\
31 & \quad\quad\quad\quad\quad $\branch_1 \leftarrow \branch_1 \cup \branch'_1$\\
32 & \quad\quad\quad\quad \EndWhile\\
33 & \quad\quad\quad\quad $\Tableau \leftarrow (\Tableau \setminus \set{\branch}) \cup \set{\branch \cup \branch_1}$\\
34 & \quad\quad\quad \EndIf\\
35 & \quad\quad \EndWhile\\
36 & \quad \EndWhile\\
37 & \quad \If{$\Tableau$ is open}\\
38 & \quad\quad \Return{\True}\\
39 & \quad \Else\\
40 & \quad \quad \Return{\False}\\
41 & \quad \EndIf\\
42 & \End\\
\end{tabular}
\caption{%
\label{tab:alg2}%
Algorithm implementing the tableaux method for $\lang$
}
\end{table}

\begin{theorem}[Termination]
\label{theo:termination_2}
The algorithm in Table~\ref{tab:alg2} halts for every input $(\Vals, \phi_0)$.
\end{theorem}

\begin{theorem}[Complexity]
\label{theo:complexity_2}
The amount of time used by the algorithm in Table~\ref{tab:alg2} is an exponential function of the length of the input $(\Vals, \phi_0)$.
\end{theorem}

Thus, the algorithm in Table~\ref{tab:alg2} works in time exponential on $\len(\phi_0)$.
This is as expected, given that the model checking problem in full $\dlpa$ is in EXPTIME \cite{HerzigLMT-Ijcai11}.
%
\section{Discussion and Conclusion}
\label{sec:conclu}
%
In this paper, we have defined a linear reduction of satisfiability checking into model checking in $\dlpa$.
We also define analytic tableaux methods for model checking formulas in the star-free fragment and in full $\dlpa$.
The complexity of these methods match the complexity class of their respective problems.
In the sequel, we compare such methods to similar approaches and discuss possible improvements and extensions.
\paragraph{Comparisons.}
The methods presented in this paper have been inspired by others already proposed in the literature.
For instance, De Giacomo and Massacci \cite{DeGiacomoMassacci00} (see also \cite{HustadtSchmidt10}) inspired the technique for the Kleene star.
As already mentioned, the naive strategy would generate tableau branches with size exponential in the length of the input formula.
The idea of keeping the tree in memory and perform equality tests comes from that work.
%
%
%
%
\paragraph{Assignments of Propositional Variables to Formulas.}
\dlpa can be extended with assignments $\alpha$ to formulas in $\lang$, instead of the simpler $\set{\top, \bot}$.
The corresponding tableau rule R$\lbox{\alpha}$ would be as follows:
\[
\begin{array}{c}
\sigma: \lbox{\alpha}\phi\\
\hline
\begin{array}{l|l|c|l}
\begin{array}{l}
\sigma: \psi_1\\
\sigma: \psi_2\\
\vdots\\
\sigma: \psi_n\\
\sigma\alpha: \prop_1\\
\sigma\alpha: \prop_2\\
\vdots\\
\sigma\alpha: \prop_n\\
\sigma\alpha: \phi
\end{array}
&
\begin{array}{l}
\sigma: \lnot\psi_1\\
\sigma: \psi_2\\
\vdots\\
\sigma\alpha: \psi_n\\
\sigma\alpha: \lnot\prop_1\\
\sigma\alpha: \prop_2\\
\vdots\\
\sigma\alpha: \prop_n\\
\sigma\alpha: \phi
\end{array}
&
\dots
&
\begin{array}{l}
\sigma: \lnot\psi_1\\
\sigma: \lnot\psi_2\\
\vdots\\
\sigma: \lnot \psi_n\\
\sigma\alpha: \lnot\prop_1\\
\sigma\alpha: \lnot\prop_2\\
\vdots\\
\sigma\alpha: \lnot\prop_n\\
\sigma\alpha: \phi
\end{array}
\end{array}
\end{array}
\]
where we assume that the domain of $\alpha$ is $\set{\prop_1, \dots, \prop_n}$ and let $\alpha(\prop_i) = \psi_i$.

In spite of the apparent complexity of this tableau rule, we believe that the complexity of the method is not affected in the star-free fragment.
For the full language, we have to include a cut rule that ranges over all sub-formulas of the input formula $\phi_0$.
The reason is to permit the equality test to work also with all formulas $\psi_i$ that are included in the tableau when the new rule R$\lbox{\alpha}$ is applied.
Again, we believe that the complexity remains the same.
\paragraph{Other $\pdl$ Connectives.}
The integration of converse, complement, intersection and other $\pdl$ program connectives is also on our agenda.
For instance, we believe that we can apply techniques similar to the ones in \cite{Nguyen_Szalas-2009-KSE,Gore_Widmann-2010-AR,AbateGW09} for the converse.
In this case though, it is not clear whether complexity (or even decidability) results remain the same.
This is subject of future work.

%
%

\bibliographystyle{plain}
\bibliography{biblio}
\onlylong{%
\newpage
\appendix
\section{Rules in Numerator-Denominator Form}
For the comfort of the reader we present here the tableau rules in the more traditional numerator-denominator form.
\begin{table}
\begin{align*}
\begin{array}{cc}
(\text{R}\lbox{\alpha})
\quad
\begin{array}[t]{l}
\sigma: \lbox{\alpha}\phi\\
\hline
\sigma\alpha: \prop_1\\
\vdots\\
\sigma\alpha: \prop_n\\
\sigma\alpha: \lnot\prop_{n+1}\\
\vdots\\
\sigma\alpha: \lnot\prop_{n+m}\\
\sigma\alpha: \phi
\end{array}
&
\qquad\qquad
(\text{R}\ldia{\alpha})
\quad
\begin{array}[t]{l}
\sigma: \lnot\lbox{\alpha}\phi\\
\hline
\sigma\alpha: \prop_1\\
\vdots\\
\sigma\alpha: \prop_n\\
\sigma\alpha: \lnot\prop_{n+1}\\
\vdots\\
\sigma\alpha: \lnot\prop_{n+m}\\
\sigma\alpha: \lnot\phi
\end{array}
\\
\\
(\text{R}\lbox{?})
\quad
\begin{array}{c}
\sigma: \lbox{\psi?}\phi\\
\hline
\begin{array}{l|l}
\sigma: \lnot\psi
&
\sigma: \phi
\end{array}
\end{array}
&
\qquad\qquad
(\text{R}\ldia{?})
\quad
\begin{array}{l}
\sigma: \lnot\lbox{\psi?}\phi\\
\hline
\sigma: \psi\\
\sigma: \lnot\phi
\end{array}
\\
\\
(\text{R}\lbox{{;}})
\quad
\begin{array}{l}
\sigma: \lbox{\pi_1; \pi_2}\phi\\
\hline
\sigma: \lbox{\pi_1}\lbox{\pi_2}\phi
\end{array}
&
\qquad\qquad
(\text{R}\ldia{{;}})
\quad
\begin{array}{l}
\sigma: \lnot\lbox{\pi_1; \pi_2}\phi\\
\hline
\sigma: \lnot\lbox{\pi_1}\lbox{\pi_2}\phi
\end{array}
\\
\\
(\text{R}\lbox{\ndet})
\quad
\begin{array}{l}
\sigma: \lbox{\pi_1 \ndet \pi_2}\phi\\
\hline
\sigma: \lbox{\pi_1}\phi\\
\sigma: \lbox{\pi_2}\phi\\
\end{array}
&
\qquad\qquad
(\text{R}\ldia{\ndet})
\quad
\begin{array}{c}
\sigma: \lnot\lbox{\pi_1 \ndet \pi_2}\phi\\
\hline
\begin{array}{l|l}
\sigma: \lnot\lbox{\pi_1}\phi
&
\sigma: \lnot\lbox{\pi_2}\phi
\end{array}
\end{array}
\\
\\
(\text{R}\lbox{\ast})
\quad
\begin{array}{l}
\sigma: \lbox{\pi^\ast}\phi\\
\hline
\sigma: \phi\\
\sigma: \lbox{\pi}\lbox{\pi^\ast}\phi
\end{array}
&
\qquad\qquad
(\text{R}\ldia{\ast})
\quad
\begin{array}{c}
\sigma: \lnot\lbox{\pi^\ast}\phi\\
\hline
\begin{array}{l|l}
\sigma: \lnot\phi
&
\sigma: \lnot\lbox{\pi}\lbox{\pi^\ast}\phi
\end{array}
\end{array}
\end{array}
\end{align*}
\caption{%
\label{tab:rules}
Tableau rules for the operator $\lbox{~}$.
In R$\lbox{\alpha}$ and R$\ldia{\alpha}$, we assume that
$\dom(\alpha) = \set{\prop_1, \dots, \prop_n, \prop_{n+1}\dots, \prop_{n+m}}$ and also
$\alpha(\prop_1) = \dots = \alpha(\prop_n) = \top$, 
and
$\alpha(\prop_{n+1}) = \dots = \alpha(\prop_{n+m}) = \bot$.
}
\end{table}
\section{Proofs}

\begin{relemma}{\ref{lem:sat_preservation_1}}{Consistency Preservation}
For each tableau rule $\rho$, if branch $\branch$ is consistent, then the set of branches $\Branches$ generated by the application of $\rho$ to $\branch$ contains a consistent branch.
\end{relemma}

\begin{proof}
The proofs for the rules R$\lnot$, R$\land$ and R$\lor$ are easy and left to the reader.
For rule R$\lbox{\alpha}$, note that, because $\branch$ is consistent, we have $\Vals^\sigma \models \lbox{\alpha}\phi$. 
Then $\Vals^{\sigma\alpha} \models \phi$ by the truth condition for $\lbox{\alpha}$.
Moreover, by the definition of updates we have:
\begin{itemize}
\item 
$\Vals^{\sigma\alpha} \models \prop$ for all $\prop \in \dom(\alpha)$ such that $\alpha(\prop) = \top$, and 
\item 
$\Vals^{\sigma\alpha} \models \lnot\prop$ for all $\prop \in \dom(\alpha)$ such that $\alpha(\prop) = \bot$.
\end{itemize}
For the remaining tableau rules, namely
R$\ldia{\alpha}$,
R$\lbox{{;}}$, R$\ldia{{;}}$, R$\lbox{\ndet}$, and R$\ldia{\ndet}$,
RP1 and RP2,
the reasoning is similar and left to the reader. 
\hfill
\qed
\end{proof}

\begin{retheorem}{\ref{theo:soundness_1}}{Soundness}
If $\Vals \models \phi_0$ then there is no closed tableau for $(\Vals, \phi_0)$.
\end{retheorem}

\begin{proof}
Assume that $\Vals \models \phi_0$.
Then the initial tableau for $(\Vals, \phi_0)$ is consistent.
It follows from Lemma~\ref{lem:sat_preservation_1} that all tableaux for $\phi_0$ have at least one consistent branch $\branch$.
Now, towards a contradiction, assume that $\branch$ is closed.
Then $\branch$ contains both $\tuple{\sigma, \psi}$ and $\tuple{\sigma, \lnot\psi}$, for some $\sigma$ and $\psi$.
However, since $\branch$ is consistent, $\Vals^\sigma \models \psi$ and $\Vals^\sigma \models \lnot\psi$, which is a contradiction.
Therefore, $\branch$ is not closed neither is the tableau containing it.
\hfill
\qed
\end{proof}

\begin{retheorem}{\ref{theo:completeness_1}}{Completeness}
If there is no closed tableau for $(\Vals, \phi_0)$ then $\Vals \models \phi_0$.
\end{retheorem}
\medskip

\begin{proof}
Suppose there is no closed tableau for $(\Vals, \phi_0)$.
Let $\branch$ be an open and saturated branch of a tableau for $(\Vals, \phi_0)$.
We prove that, for every pair $\tuple{\sigma, \psi} \in \branch$, we have $\Vals^\sigma \models \psi$.
The proof is done by induction on $\len(\sigma) + \len(\psi)$ and, in particular, establishes that $\Vals \models \phi_0$, since $\tuple{\emptyseq, \phi_0} \in \branch$.
\medskip

Induction base:
We consider two cases:
\begin{itemize}
\item
Let $\sigma = \emptyseq$ and $\psi = \prop \in \Props$.
Then $\Vals \models \prop$, otherwise $\branch$ would be closed since $\branch_0 \subseteq \branch$.
\item
Let $\sigma = \emptyseq$ and $\psi = \lnot\prop$.
Then $\Vals \models \lnot\prop$, otherwise $\branch$ would be closed since $\branch_0 \subseteq \branch$.
\end{itemize}
\medskip

Induction Hypothesis:
For every $\tuple{\sigma, \psi} \in \branch$,
if $\len(\sigma) + \len(\psi) \leq n$,
then $\Vals^\sigma \models \psi$.
\medskip

Induction step:
Let $\len(\sigma) + \len(\psi) = n + 1$.
We only give some of all possible cases:
\begin{itemize}
\item
Let $\sigma = \sigma_1\alpha$ and $\psi = \prop \in \Props$.
We consider two sub-cases:
\begin{itemize}
\item
Let $\prop \not\in \dom(\alpha)$.
We have $\tuple{\sigma_1, \lnot\prop} \not\in \branch$,
otherwise $\branch$ would be closed, because it is saturated under RP2.
Then we have $\tuple{\sigma_1, \prop} \in \branch$,
because the branch is saturated under RP1 and $\branch_0 \subseteq \branch$.
By induction hypothesis, we have $\Vals^{\sigma_1} \models \prop$.
Since $\prop \not\in \dom(\alpha)$, we also have $\Vals^{\sigma_1\alpha} \models \prop$.
\item
Let $\prop \in \dom(\alpha)$.
We then must have $\alpha(\prop) = \top$:
otherwise $\branch$ would not only contain $\tuple{\sigma_1\alpha, \prop}$, 
but also $\tuple{\sigma_1\alpha, \lnot\prop}$ 
(by the application of rule R$\lbox{\alpha}$)
and $\branch$ would therefore be closed.
Hence, by the definition of updates $\prop \in \Vals^{\sigma_1\alpha}$.
The latter means that $\Vals^{\sigma_1\alpha} \models \prop$.
\end{itemize}
\item
Let $\sigma = \sigma_1\alpha$ and $\psi = \lnot\prop$.
Again, we consider two sub-cases:
\begin{itemize}
\item
Let $\prop \not\in \dom(\alpha)$.
We have $\tuple{\sigma_1, \prop} \not\in \branch$:
otherwise, $\branch$ would be closed, since it is saturated under RP1.
Then we have $\tuple{\sigma_1, \lnot\prop} \in \branch$, because the branch is saturated under RP2 and $\branch_0 \in \branch$.
Then $\prop \not\in \Vals^{\sigma_1}$ (by the induction hypothesis) and thus $\prop \not\in \Vals^{\sigma_1\alpha}$.
Then $\Vals^{\sigma_1\alpha} \models \lnot\prop$.
\item
Let $\prop \in \dom(\alpha)$.
Note that we have $\alpha(\prop) = \bot$:
otherwise $\branch$ would be closed, because it would contain $\tuple{\sigma_1\alpha, \prop}$ and $\tuple{\sigma_1\alpha, \lnot\prop}$, since it is saturated under R$\lbox{\alpha}$.
Then $\prop \not\in \Vals^{\sigma_1\alpha}$
(by its definition)
Then $\Vals^{\sigma_1\alpha} \models \lnot\prop$.
\end{itemize}
\item
Let $\psi = \lnot\lnot\psi_1$.
If $\tuple{\sigma, \lnot\lnot\psi_1} \in \branch$
then $\tuple{\sigma, \psi_1} \in \branch$
(because $\branch$ is saturated under R$\lnot$).
By Induction Hypothesis we have $\Vals^\sigma \models \psi_1$.
Therefore $\Vals^\sigma \models \lnot\lnot\psi_1$ by the truth condition for negation.
\item
Let $\psi = \psi_1 \land \psi_2$.
If $\tuple{\sigma, \psi_1 \land \psi_2} \in \branch$
then $\tuple{\sigma, \psi_1}, \tuple{\sigma, \psi_2} \in \branch$
(because $\branch$ is saturated under R$\land$).
By Induction Hypothesis we have 
$\Vals^\sigma \models \psi_1$ and $\Vals^\sigma \models \psi_2$.
Therefore $\Vals^\sigma \models \psi_1 \land \psi_2$
by the truth condition for conjunction.
\item
Let $\psi = \lnot(\psi_1 \land \psi_2)$.
If $\tuple{\sigma, \lnot(\psi_1 \land \psi_2)} \in \branch$
then $\tuple{\sigma, \lnot\psi_1} \in \branch$ or $\tuple{\sigma, \lnot\psi_2} \in \branch$
(because $\branch$ is saturated under R$\lor$).
By Induction Hypothesis we have 
$\Vals^\sigma \models \lnot\psi_1$ or $\Vals^\sigma \models \lnot\psi_2$.
Therefore $\Vals^\sigma \models \lnot(\psi_1 \land \psi_2)$
by the truth conditions for negation and conjunction.
\item
%
Let $\psi = \lbox{\alpha}\psi_1$.
If $\tuple{\sigma, \lbox{\alpha}\psi_1} \in \branch$
then $\tuple{\sigma\alpha, \psi_1} \in \branch$
(because $\branch$ is saturated under R$\lbox{\alpha}$).
Then, $\Vals^{\sigma\alpha} \models \psi_1$
(by Induction Hypothesis, because $\len(\sigma) + \len(\alpha) + \len(\psi_1) < \len(\sigma) + \len(\lbox{\alpha}\psi_1) = \len(\sigma) + 1 + \len(\alpha) + \len(\psi_1)$).
Therefore, $\Vals^\sigma \models \lbox{\alpha}\psi_1$
(by definition).
%
%
\item
Let $\psi = \lbox{\psi_1?}\psi_2$.
If $\tuple{\sigma, \lbox{\psi_1?}\psi_2} \in \branch$
then, because $\branch$ is saturated under rule R$\lbox{?}$, we consider two sub-cases.
Either
(1) $\tuple{\sigma, \lnot\psi_1} \in \branch$
or
(2) $\tuple{\sigma, \psi_2} \in \branch$.
In both sub-cases, we have $\Vals^\sigma \models \psi_1$ implies $\Vals^\sigma \models \psi_2$
(by Induction Hypothesis).
Therefore, $\Vals^\sigma \models \lbox{\psi_2?}\psi_2$
(by definition).
%
\item
Let $\psi = \lbox{\pi_1; \pi_2}\psi_1$.
If $\tuple{\sigma, \lbox{\pi_1; \pi_2}\psi_1} \in \branch$
then $\tuple{\sigma, \lbox{\pi_1}\lbox{\pi_2}\psi_2} \in \branch$
(because $\branch$ is saturated under rule R$\lbox{{;}}$).
Then it is easy to see that $\tuple{\sigma\sigma_1, \lbox{\pi_2}\psi_1} \in \branch$,
for all execution traces $\sigma_1 \in \seq(\pi_1)$
Then $\Vals^{\sigma\sigma_1} \models \lbox{\pi_2}\psi_1$
(by Induction Hypothesis,
since $\len(\sigma) + \len(\sigma_1) + 1 + \len(\pi_2) + \len(\psi_1) = \len(\sigma\sigma_1) + \len(\lbox{\pi_2}\psi_1) <
\len(\sigma) + \len(\lbox{\pi_1; \pi_2}\psi_1) = \len(\sigma) + 1 + \len(\pi_1) + 1 + \len(\pi_2) + \len(\psi_1)$).
The latter means that $\Vals^\sigma \models \lbox{\pi_1; \pi_2}\psi_1$.
\item
Let $\psi = \lbox{\pi_1 \ndet \pi_2}\psi_1$.
If $\tuple{\sigma, \lbox{\pi_1 \ndet \pi_2}\psi_1} \in \branch$
then $\tuple{\sigma, \lbox{\pi_1}\psi_1}, \tuple{\sigma, \lbox{\pi_2}\psi_1} \in \branch$
(because $\branch$ is saturated under rule R$\lbox{\ndet}$).
Then, $\Vals^\sigma \models \lbox{\pi_1}\psi_1$ and $\Vals^\sigma \models \lbox{\pi_2}\psi_1$
(by Induction Hypothesis).
Therefore, $\Vals^\sigma \models \lbox{\pi_1 \ndet \pi_2}\psi_1$
(by definition).
\item
The cases where $\psi = \lnot\lbox{\pi}\psi_1$ are analogous to the last ones.
\hfill
\qed
\end{itemize}
\end{proof}

\begin{retheorem}{\ref{theo:termination_1}}{Termination}
The algorithm in Table~\ref{tab:alg1} halts for every input $(\Vals, \phi_0)$.
\end{retheorem}

\begin{proof}
It is enough to show that function \Satsf\ is eventually called with an argument $\branch$ which is either a closed or a saturated branch.
Assume that, during the execution, branch $\branch$ is passed as argument to a call of function \Satsf.
Assume that $\branch$ contains a witness $\lambda$ to one of the tableau rules.
Then the function will be called recursively with a new branch $\branch_1$ wherein $\lambda$ is marked `non-applicable', so it will never be a witness again.
Moreover, $\branch_1$ differs from $\branch$ by some additional labelled formulas that are shorter than $\lambda$.
Therefore, by an easy induction on the length of labelled formulas, we show that function \Satsf\ will eventually generate a branch $\branch_1$ which is either closed or saturated.
The details are omitted.
\hfill
\qed
\end{proof}

\begin{retheorem}{\ref{theo:complexity_1}}{Complexity}
The amount of memory used by the algorithm in Table~\ref{tab:alg1} is a polynomial function of the length of the input $(\Vals, \phi_0)$.
\end{retheorem}

\begin{proof}
Each call of function \Satsf\ generates a new tableau branch.
This branch remains in memory during the recursive calls and is released once the present call of the function finishes its execution returning \True\ or \False.
Therefore, to prove our claim, it is enough to show that the amount of memory used by each tableau branch is a polynomial function of $\len(\phi_0)$ and that the number of successive recursive calls to \Satsf\ is a polynomial function of $\len(\phi_0)$ as well.

First, we observe that the initial branch $\branch_0$ contains only formulas from $\ecl(\phi_0)$.

Second, each time \Satsf\ is called with branch $\branch$ as argument, all the labelled formulas in $\branch$ have the same label.
Since the amount of memory used by a branch is bounded by the number of different labelled formulas it contains,
it then follows from Lemma~\ref{lem:cl} that the number of different labelled formulas in $\branch$ is bounded by $2\len(\phi_0)$.

Third, the number of successive recursive calls during the local saturation of the tableau is bounded by the number of different labelled formulas a successor may contain.
This number is $2\len(\phi_0)$, again by Lemma~\ref{lem:cl}.
Now, recall that the list of successors created by the algorithm during successive recursive calls of \Satsf\ corresponds to one execution trace from input formula $\phi_0$.
The length of each execution trace is bounded by $\len(\phi_0)$, by Lemma~\ref{lem:exe}.
Then the number of successive recursive calls that create new successors is bounded by $\len(\phi_0)$.
Finally, the total number of successive recursive calls to \Satsf\ is bounded by $2\len(\phi_0)^2$.

We then conclude that the amount of memory used by the algorithm is proportional to $4\len(\phi_0)^3$.
\hfill
\qed
\end{proof}

\begin{relemma}{\ref{lem:sat_preservation_2}}{Satisfiability Preservation}
For each tableau rule $\rho$, if branch $\branch$ is consistent, then the set of branches $\Branches$ generated by the application of $\rho$ to $\branch$ contains a consistent branch.
\end{relemma}

\begin{proof}
For the rules that are already part of in the method for star-free \dlpa
the proof is the same as in the proof of Lemma~\ref{lem:sat_preservation_1}.
For the other cases, we have:
\begin{itemize}
\item
Rule R$\lbox{\ast}$:
If $\Vals^\sigma \models \lbox{\pi^\ast}\phi$ then, by Proposition~\ref{pro:principles}, 
$\Vals^\sigma \models \phi$ and $\Vals^\sigma \models \lbox{\pi}\lbox{\pi^\ast}\phi$.
\item
Rule R$\ldia{\ast}$:
If $\Vals^\sigma \models \lnot\lbox{\pi^\ast}\phi$
then $\Vals^\sigma \not\models \lbox{\pi^\ast}\phi$, and the latter is the case
iff
$\Vals^\sigma \models \lnot\phi$ or $\Vals^\sigma \models \lnot\lbox{\pi}\lbox{\pi^\ast}\phi$,
again due to Proposition~\ref{pro:principles}.
\hfill
\qed
\end{itemize}
\end{proof}

\begin{retheorem}{\ref{theo:soundness_2}}{Soundness}
If $\Vals \models \phi_0$ then there is no closed tableau for $(\Vals, \phi_0)$.
\end{retheorem}

\begin{proof}
Assume that $\Vals \models \phi_0$.
Then, the initial tableau for $(\Vals, \phi_0)$ is consistent.
It follows from Lemma~\ref{lem:sat_preservation_2} that all tableaux for $(\Vals, \phi_0)$ have at least one consistent branch $\branch$.
Now, towards a contradiction, assume that $\branch$ is closed.
Then, either
(1) $ \branch$ contains both $\tuple{\sigma, \phi}$ and $\tuple{\sigma, \lnot\phi}$, 
for some $\sigma$ and $\phi$;
or
(2) $\branch$ is saturated and contains an unfulfilled eventuality 
$\tuple{\sigma, \lnot\lbox{\pi^\ast}\phi}$.
In the first case, (because $\branch$ is consistent) $\Vals^\sigma \models \phi$ and $\Vals^\sigma \models \lnot\phi$, which is a contradiction.
In the second case, (again because $\branch$ is consistent) $\Vals^\sigma \models \lnot\lbox{\pi^\ast}\phi$.
Moreover,
$\branch$ contains $\tuple{\sigma\sigma', \phi}$, for all execution traces $\sigma' \in \seq(\pi^\ast)$, by the saturation of R$\ldia{\ast}$ and because the eventuality is not fulfilled.
Then, $\Vals^{\sigma\sigma'} \models \phi$, for all execution traces $\sigma' \in \seq(\pi^\ast)$
(because the branch is consistent).
Then, $\Vals^\sigma \models \lbox{\pi^n}\phi$, for all $n \geq 0$.
The latter implies $\Vals^\sigma \models \lbox{\pi^\ast}\phi$, which contradicts the hypothesis.
So $\branch$ is not closed, and therefore the tableau containing $\branch$ cannot be closed.
\hfill\qed
\end{proof}

\begin{retheorem}{\ref{theo:completeness_2}}{Completeness}
If there is no closed tableau for $(\Vals, \phi_0)$ then $\Vals \models \phi_0$.
\end{retheorem}

\begin{proof}
The proof is essentially the same as for Theorem~\ref{theo:completeness_1}.
We only add the induction step case for the Kleene star operator here:
\begin{itemize}
\item[-]
Let $\psi = \lbox{\pi^\ast}\psi_1$.
If $\tuple{\sigma, \lbox{\pi^\ast}\psi_1} \in \branch$ then $\tuple{\sigma\sigma', \psi_1} \in \branch$,
for all execution traces $\sigma' \in \seq(\pi^\ast)$
(because $\branch$ is saturated, in particular, under rule R$\lbox{\ast}$).
Then, $\Vals^{\sigma\sigma'} \models \psi$, for all execution traces $\sigma' \in \seq(\pi^\ast)$
(by Induction Hypothesis),
iff
$\Vals^\sigma \models \lbox{\pi^n}\psi$, for all $n \in \Nats_0$,
iff
$\Vals^\sigma \models \lbox{\pi^\ast}\psi$.
\end{itemize}
For the case where $\psi = \lnot\lbox{\pi^\ast}\psi$ we use the fact that the branch $\branch$ is not closed, which means that the eventuality is fulfilled in $\branch$, by definition.
\hfill\qed
\end{proof}

\begin{retheorem}{\ref{theo:termination_2}}{Termination}
The algorithm in Table~\ref{tab:alg2} halts for every input $(\Vals, \phi_0)$.
\end{retheorem}

\begin{proof}
It is enough to show that the algorithm eventually generates a tableau such that all its branches are either closed or saturated.
The algorithm has two parts: local saturation and successor creation.

First, assume that the latest generated tableau $\Tableau$ contains an open and unsaturated branch $\branch$ with a witness $\lambda$ to one of the tableau rules of the local saturation part. 
Then the algorithm updates $\Tableau$ by marking $\lambda$ as `non-applicable', so it will never be a witness again.
Moreover, the new branches of the updated tableau $\Tableau$ differ from the old ones by somme additional labelled formulas that are either shorter than $\lambda$ or (in the case of rules R$\lbox{\ast}$ and R$\ldia{\ast}$) that can no longer be witnesses to these rules any more.
Therefore, by an easy induction on the length of labelled formulas, we show that the algorithm will eventually generate a tableau such that all its branches are either closed or saturated for these rules.
The details are omitted.

Second, assume that the latest generated tableau $\Tableau$ contains an open and unsaturated branch $\branch$ with a witness $\lambda$ to one of the tableau rules R$\lbox{\alpha}$ and R$\ldia{\alpha}$.
Then it marks $\lambda$ as `non-applicable' and the updated $\Tableau$ contains new branches with somme additional labelled formulas $\tuple{\sigma, \psi}$, where $\sigma$ is a new label and $\psi \in \ecl(\phi_0)$.
Since $\ecl(\phi_0)$ is finite, there cannot be an infinite number of different labelled formulas whose labels different from another label of the tableau.
Thus, the equality test will eventually succeeds and new successors won't be created indefinitely.
\hfill
\qed
\end{proof}

\begin{retheorem}{\ref{theo:complexity_2}}{Complexity}
The amount of time used by the algorithm in Table~\ref{tab:alg2} is an exponential function of the length of the input $(\Vals, \phi_0)$.
\end{retheorem}

\begin{proof}
The amount of time used by the algorithm in Table~\ref{tab:alg2} is bounded by the number of rule applications during the execution and the time spent on the equality tests.

First, for each successor, the local saturation part performs at most $2\len(\phi_0)$ rule applications, because it is the maximum size of $\ecl(\phi_0)$ (by Lemma~\ref{lem:cl}).
Second, there can be at most $2^{2\len(\phi_0)}$ different labels in the entire tableau $\Tableau$, because it is the maximum size of $\mathcal{P}(\ecl(\phi_0))$.
Then the successor creation part can generate at most $2^{2\len(\phi_0)}$ different labels until the equality test succeeds.
Moreover, the equality test itself takes time proportional to $2^{2\len(\phi_0)}$, by the same reasons.

Overall, the amount of time used by the algorithm is bounded by $2\len(\phi_0) \times 2^{2\len(\phi_0)} \times 2^{2\len(\phi_0)}$ equals to $2^{4\len(\phi_0)+1}\len(\phi_0)$, which is an exponential function of the length of the input formula $\phi_0$.
\hfill
\qed
\end{proof}
}
\end{document}